\documentclass[letterpaper, 10 pt, conference]{IEEEtran}
\usepackage[]{graphicx}
\graphicspath{{figures/}}
\DeclareGraphicsExtensions{.eps,.pdf}

\usepackage{amsmath,amssymb,amsthm,dsfont}
\usepackage[mathscr]{euscript}
\usepackage{mathtools}
\usepackage{verbatim}
\usepackage{stmaryrd}
\DeclarePairedDelimiter \norm{\lVert}{\rVert}%
\newtheorem{theorem}{Theorem}%[section]

\newtheorem{lemma}{Lemma}
\theoremstyle{definition}
\theoremstyle{remark}

\allowdisplaybreaks[4]
\usepackage{multirow}
\usepackage{color}
\usepackage{ctable}

\usepackage{cite}
\usepackage[
colorlinks=true,
citecolor=black,
linkcolor=black,
urlcolor=black,
]{hyperref}

\usepackage[symbol]{footmisc}

\IEEEoverridecommandlockouts

\begin{document}	
	\title{\vspace*{0.25in}Joint Coordination-Channel Coding  for Strong Coordination over Noisy Channels\\ Based on Polar Codes}
	\author{\thanks{This work is supported by NSF grants CCF-1440014, CCF-1439465.}
		\IEEEauthorblockN{Sarah A. Obead, J\"{o}rg Kliewer}
		\IEEEauthorblockA{Department of Electrical and Computer Engineering\\
			New Jersey Institute of Technology\\
			Newark, New Jersey 07102\\
			Email:sao23@njit.edu, jkliewer@njit.edu\\}
		\and %\vspace{3 ex}
		\IEEEauthorblockN{ Badri N. Vellambi}
		\IEEEauthorblockA{Research School of Computer Science\\
			Australian National University\\
			Acton, Australia 2601\\
			Email: badri.n.vellambi@ieee.org}
	}
	
	\maketitle
	\thispagestyle{empty}
	\pagestyle{empty}

%%%%%%%%%%%%%%%%%%%%%%%%%%%%%%%%%%%%%%%%%%%%%%%%%%%%%%%%%%%%%%%%%%%%%%%%%%%%%%%%%%%%%%%%%%%%%%%%%%%%%%%%%%%%%%%%%%%%%%%%%
	\begin{abstract}
		We construct a joint coordination-channel polar coding
                scheme for strong coordination of actions between two agents
                $\mathsf X$ and $\mathsf Y$, which communicate over a
                discrete memoryless channel (DMC) such that the joint
                distribution of actions follows a prescribed probability
                distribution. We show that polar codes are able to achieve
                our previously established  inner bound to the strong noisy coordination capacity
                region  and thus provide a constructive alternative to a random coding proof. Our polar coding scheme also offers a constructive solution to a channel simulation problem where a DMC and shared randomness are together employed to simulate another DMC. In particular, our proposed solution is able to utilize the randomness of the DMC to reduce the amount of local randomness required to generate the sequence of actions at agent $\mathsf Y$. By leveraging our earlier random coding results for this problem, we conclude that the proposed joint coordination-channel coding scheme strictly outperforms a separate scheme in terms of achievable communication rate for the same amount of injected randomness into both systems.
	\end{abstract}
	\vspace{-0.5ex}
	\section{Introduction}
	\vspace{-0.5ex}
	A fundamental problem in decentralized networks is to coordinate
	activities of different agents with the goal of reaching a state of
	agreement. Such a problem arises in a multitude of applications,
	including networks of autonomous robots, smart traffic control, and
	distributed computing problems. For such applications, coordination is understood to be the ability to arrive at a prescribed joint distribution of actions at all agents in the network. In information theory, two different notions of coordination are explored: (i) empirical coordination, which only requires the
	normalized histogram of induced joint actions to approach a desired
	target distribution, and (ii) strong coordination, where the
	sequence of induced joint actions must be statistically close (i.e., nearly
	indistinguishable) from a given target probability mass function (pmf).  
	
	A significant amount of work has been devoted to finding the
	capacity regions of various coordination problems based on both
	empirical and strong coordination
	\cite{soljanin2002compressing,cuff2010coordination,gohari2011generating,haddadpour2012coordination,cuff2013distributed,bereyhi2013empirical,bloch2013strong,BK14,yassaee2015channel,VKB16},
	where
	\cite{haddadpour2012coordination,bereyhi2013empirical,bloch2013strong,BK14,VKB16} focus on
	small to moderate network settings. 

While all these works address the noiseless case,
	coordination over noisy channels has received only little attention in
	the literature so far. However, notable exceptions are
	\cite{CS11,HYBGA13,GuliaISIT}. For example, in \cite{CS11}  joint empirical coordination of the channel
	inputs/outputs of a noisy communication channel with source and
	reproduction sequences is considered. Also, in \cite{HYBGA13} the
	notion of strong coordination is used to simulate a discrete
	memoryless channel via another channel. Recently, \cite{GuliaISIT}
        explored the strong coordination variant of the problem investigated
        in \cite{CS11}.
% by coordinating both the channel input and the
%        reproduction sequence with the source.% that was investigated in \cite{CS11}.

	As an alternative to the impracticalities of random coding, solutions for empirical and strong coordination problems have been proposed based on low-complexity polar-codes introduced by Arikan \cite{ChPolarztion2009Arikan,SPolarztion2010Arikan}. For example, polar coding for
	strong point-to-point coordination is addressed in
	\cite{bloch2012strong,CBJ16}, and empirical coordination
	for cascade networks in \cite{blasco2012polar}, respectively. The
	only existing design of polar codes for the
	noisy empirical coordination case %has been recently published in
	\cite{PCempirical2016} is based on the joint source-channel
	coordination approach in \cite{CS11}. However, to
	the best of our knowledge, polar code designs
	for noisy \emph{strong} coordination have not been proposed in the literature.
	
	In this work we consider the point-to-point coordination setup depicted in Fig.~\ref{fig:P2PCoordination} where only source and reproduction sequences are coordinated via a suitable polar coding scheme over DMCs. In particular, we design an explicit low-complexity nested polar coding scheme for strong coordination over noisy channels that
	achieves the inner bound of the two-node network capacity region of
	our earlier work 
	\cite{obead17:_stron_coord_noisy_chann}. In this work, we show
	that a joint coordination-channel coding scheme is able to strictly
	outperform a separation-based scheme in terms of achievable
	communication rate if the same amount of randomness is injected into the system. Note that our proposed  joint
	coordination-channel polar coding scheme employs nested
	codebooks similar to the polar codes for the broadcast
	channel \cite{PC_BCGoela2015}. Further, %as a byproduct,
	our polar coding scheme also offers a constructive solution to a channel simulation problem where a DMC is employed to simulate another DMC in the presence of shared randomness~\cite{HYBGA13}.

	The remainder of the paper is organized as follows. Section~II
	introduces the notation, the model under investigation, and a random
	coding construction. Section~III provides our proposed joint
	coordination-channel coding design and a proof to show that this
	design achieves the random coding inner bound.

	\vspace{-1.5ex}
	\section{Problem Statement}	
	\vspace{-0.8ex}
	\subsection{Notation}
	Let $N\triangleq2^n, n\in \mathbb{N}$. We denote the source
        polarization transform as  $G_n=RF^{\otimes n},$ where $R$ is the bit-reversal mapping defined in \cite{ChPolarztion2009Arikan}, 
	$F =\left[\begin{smallmatrix}
	1 &  0\\
	0 &  0
	\end{smallmatrix}\right],$
	and $F^{\otimes n}$ denotes the $n$-{th} Kronecker power of $F.$
	Given $X^{1:N}\triangleq({X^1,X^2,\ldots,X^N})$ and ${\cal A}\subset  \llbracket 1,N \rrbracket$, we let $X^N[{\cal A}]$ denote the components $X^i$ such that $i \in {\cal A}.$ Given two distributions $P_X(x)$ and $Q_X(x)$ defined over an alphabet
	$\cal{X},$ we let
	$\mathbb{D}(P_X(x)||Q_X(x))$ and  $\norm{P_X(x)-Q_X(x)}_{{\scriptscriptstyle TV}} $ denote the Kullback-Leibler (KL) divergence and the total variation, respectively. Given a pmf $P_X(x)$ we let $\min^{*}(P_X) = \min \, \{P_X(x): P_X(x)>0\}$.
	
	\subsection{System Model}
	The point-to-point coordination setup considered in this work is depicted
	in Fig.~\ref{fig:P2PCoordination}. Node $\mathsf X$ receives a sequence of actions
	$X^N \in \mathcal{X}^N$ specified by nature where $X^N$ is i.i.d.~according
	to a pmf $p_X$. Both nodes have access to shared randomness $J$ at rate
	$R_o$ bits/action from a common source, and each node possesses local
	randomness $M_{\ell}$ at rate $\rho_{\ell}$, ${\ell}=1,2$. 
	
	\vspace{-3ex}
	\begin{figure}[h!]
		\centering
		{\includegraphics[scale=0.56]{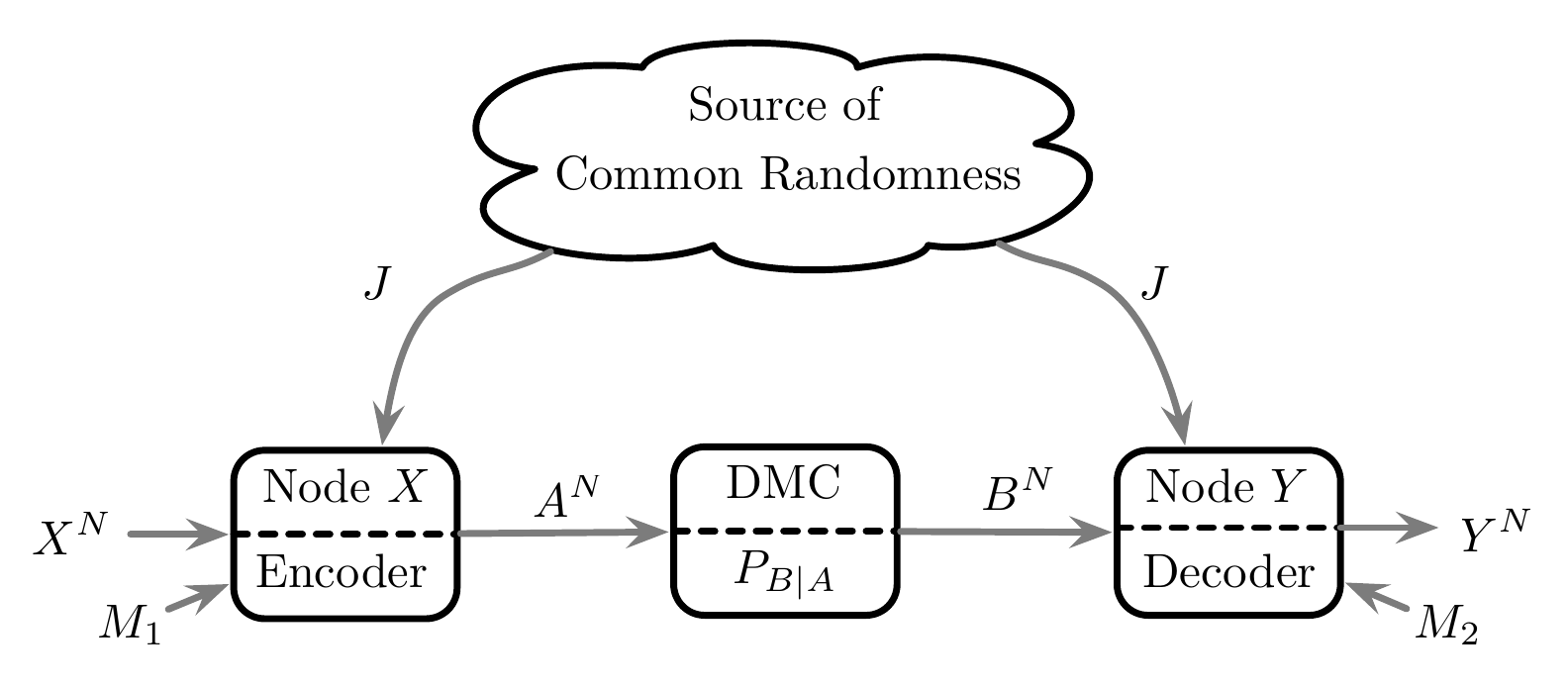}}
		\vspace{-3ex}
		\caption{Point-to-point strong coordination over a DMC.}
		\vspace{-1ex}
		\label{fig:P2PCoordination}
	\end{figure} 
	
	We wish to communicate a codeword $A^N$ corresponding to the coordination message over the rate-limited DMC $P_{B|A}$ to Node $\mathsf Y$. The \emph{codeword} $A^N$ is
	constructed based on the input action sequence $X^N$, the local randomness
	$M_1$ at Node $\mathsf X$, and the common randomness $J$. Node $\mathsf Y$ generates a
	sequence of actions $Y^N\in \mathcal{Y}^N$ based on the received codeword
	$B^N$, common randomness $J$, and local randomness $M_2$. 
	
	By assumption, the common randomness is independent
	of the action specified at Node $\mathsf X$. A strong coordination coding scheme with rates $(R_c, R_o, \rho_1,\rho_2)$ is deemed \emph{achievable} if for each $\epsilon>0$, there exists an $N\in\mathbb{N}$ such that the joint
	pmf of actions $\tilde{P}_{X^N,Y^N}$ induced by this scheme and the $N$ i.i.d. copies of desired joint pmf $(X,Y)\sim q_{XY},$ 
	$Q_{X^NY^N},$ are \emph{close} in total variation, i.e.,
	\begin{equation}\label{eq:StrngCoorCondtion}
	\norm{\tilde{P}_{X^NY^N}-Q_{X^NY^N}}_{{\scriptscriptstyle TV}} 
	<\epsilon.
	\end{equation}
	
	\subsection{Random Coding Construction}\label{sec:RandomCoding}
	Consider auxiliary random variables ${A\in\mathcal{A} \text{ and } C\in\mathcal{C}}$
	with ${(A,C)\sim P_{AC}}$ be jointly correlated with $(X,Y)$ as
	$P_{XYABC}=P_{AC}P_{X|AC}P_{B|A}P_{Y|BC}$. The joint strong
	coordination-channel random code with parameters $(R_c,R_o,R_a,N)$ \cite{obead17:_stron_coord_noisy_chann}, where ${\cal I} \triangleq \llbracket 1,2^{NR_c} \rrbracket$, ${\cal J} \triangleq\llbracket 1,2^{NR_o} \rrbracket$, and ${\cal K} \triangleq \llbracket 1,2^{NR_a} \rrbracket$, consists of
	\begin{enumerate}	
		\item Nested codebooks: A codebook $\mathscr{C}$ of size $2^{N(R_o+R_c)}$ is generated
		i.i.d.~according to a pmf $P_{C}$, i.e.,
		$C^N_{ij}\sim \prod_{l=1}^{N} P_{C}(\cdot)$ for all $(i,j) \in \cal{I}\times \cal{J}$. A codebook $\mathscr{A}$ is generated by randomly selecting $A^N_{ijk}\sim  \prod_{l=1}^{N}  P_{A|C}(\cdot|{C_{ij}^N})$ for all $(i,j,k) \in \cal{I}\times\cal{J}\times\cal{K}$. 
		\item Encoding functions:\\
		$C^N\!:   \llbracket 1,2^{NR_c} \rrbracket \!\times\! \llbracket 1,2^{NR_o} \rrbracket \! \rightarrow \mathcal{C}^N$,\\
		$A^N\!:   \llbracket 1,2^{NR_c} \rrbracket\!\times\! \llbracket 1,2^{NR_o} \rrbracket\!\times\!\llbracket 1,2^{NR_a} \rrbracket\! \rightarrow \mathcal{A}^N$.
		\item The indices $I,J,K$ are independent and uniformly distributed over $\cal{I}$,
		$\cal{J}$, and $\cal{K}$, respectively. 
		These indices select the pair of codewords $C^N_{IJ}$ and $A^N_{IJK}$ from codebooks $\mathscr{C}$ and $\mathscr{A}$.
		\item The selected codeword $A^N_{IJK}$ is sent
		through the communication DMC $P_{B|A}$, whose output $B^N$ is used to decode codeword $C^N_{\hat{I}J}$, and both are then passed
		through a DMC $P_{Y|BC}$ to obtain $Y^N$. 
	\end{enumerate} 
	The corresponding scheme is displayed in Fig.~\ref{fig:StrongCoordinationAllied}. 
	
	\begin{figure}[h!]
		\hspace{-1ex}\centering
		{\includegraphics[scale=0.51]{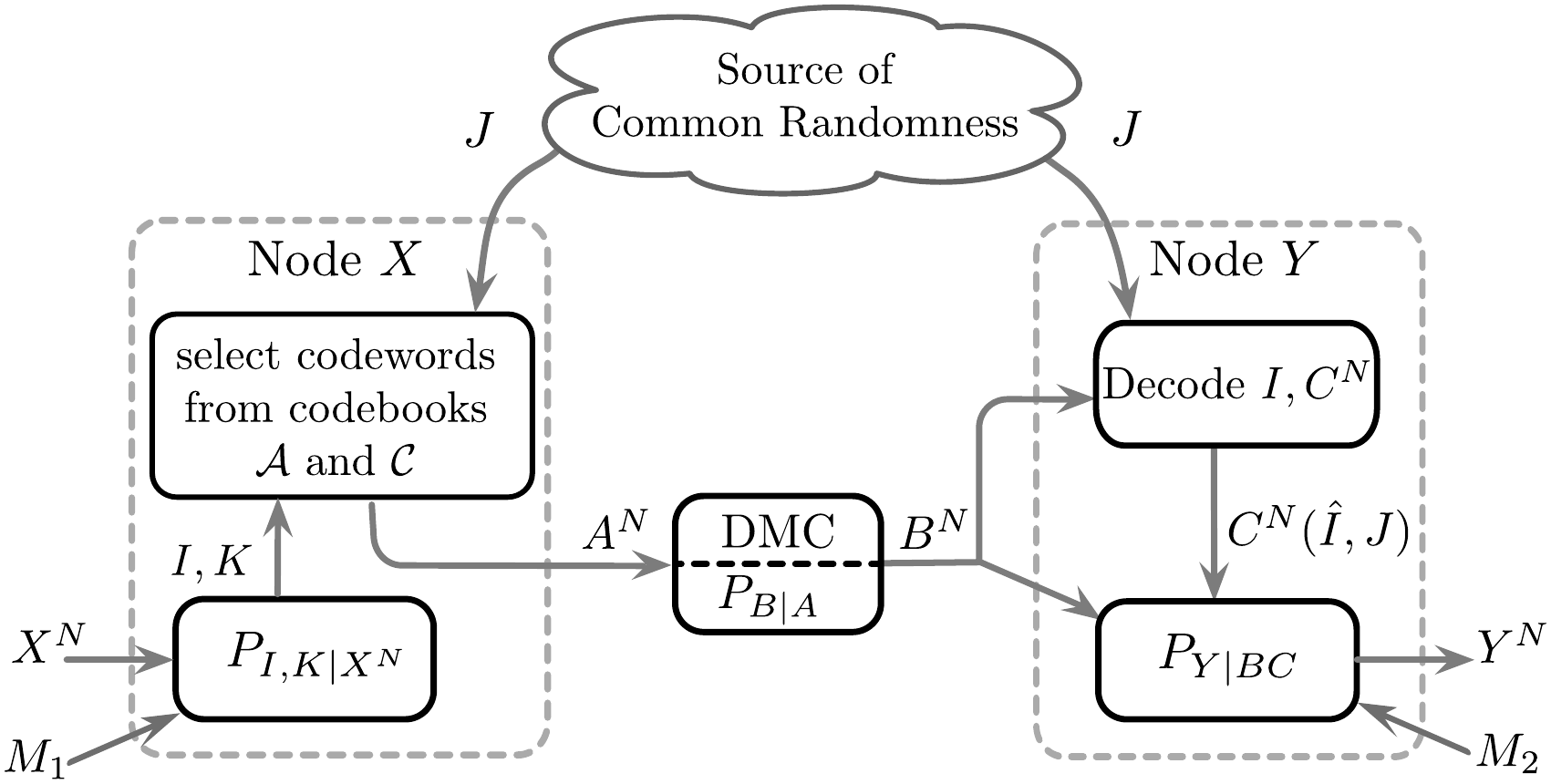}}
		\vspace{-3ex}
		\caption{Joint strong coordination-channel coding scheme.} 
		\vspace{-1ex}
		\label{fig:StrongCoordinationAllied}
	\end{figure}
	
	The following theorem provides the inner bound for strong coordination region achieved by such joint coordination-channel code.

	\begin{theorem}(Strong coordination inner bound \cite{obead17:_stron_coord_noisy_chann})\label{Thm:JointCRR}
		A tuple $(R_o, \rho_1,\rho_2)$ is achievable 
		for the strong noisy communication setup in
		Fig.~\ref{fig:P2PCoordination} if for some $R_a,R_c\geq0$,
		\begin{subequations}\label{equ:JSRR}
			\begin{align} 
			R_a+R_o+R_c & {\;>\;} I(XY;AC), \label{equ:JSRR1}\\
			R_o+R_c &{\;>\;} I(XY;C),\label{equ:JSRR2}\\
			R_a+R_c &{\;>\;} I(X;AC),\label{equ:JSRR3}\\
			R_c &{\;>\;} I(X;C),\label{equ:JSRR4}\\
			R_c &<I(B;C),\label{equ:JSRR5}\\
			\rho_1 &{\;>\;} R_a+R_c-I(X;AC),\label{equ:JSRR6}\\
			\rho_2 &{\;>\;} H(Y|BC).\label{equ:JSRR7} 
			\end{align}
		\end{subequations}
	\end{theorem}
	The underlying proofs and details of the coding mechanism for this
        joint coordination-channel coding scheme for noisy strong
        coordination are based on a complex channel resolvability framework
        \cite{obead17:_stron_coord_noisy_chann}. Channel resolvability has
        been successfully used to study different strong coordination
        problems due to its ability to approximate channel output statistics with random codebooks \cite{han1993approximation}.	
	We now propose a scheme based on polar codes that achieves the inner bound stated by Theorem~\ref{Thm:JointCRR} for the strong coordination region as follows. 
	
	\section{Nested Polar Code for Strong Coordination over Noisy Channels}
	Since the proposed joint coordination-channel coding scheme is based
        on a channel resolvability framework, we adopt the channel
        resolvability-based polar construction for noise-free strong
        coordination \cite{CBJ16} in combination with polar coding for the
        degraded broadcast channel \cite{PC_BCGoela2015}.	
	\subsection{Coding Scheme}
	Consider the random variables $X,Y,A,B,C,\widehat{C}$ distributed according to
	$Q_{XYABC\widehat{C}}$ over
	$\mathcal{X}\times\mathcal{Y}\times\mathcal{A}\times\mathcal{B}\times\mathcal{C}$
	such that $X-(A,C)-(B,\widehat{C})-Y$. Assume that $|A|=2$ and the
	distribution $Q_{XY}$ is achievable with $|C|=2$ \footnote[2]{For the
		sake of exposition, we only focus on the set of joint
		distributions over $\mathcal{X}\times\mathcal{Y}$ that are
		achievable with binary auxiliary random variables $C, A$, and over
		a binary-input DMC. The scheme can be generalized to non-binary $C, A$ with non-binary polar codes in a straightforward way~\cite{STA09}.}. Let $N\triangleq 2^n$. We describe the polar coding scheme as follows:  
	
		\begin{figure}[htb]
		\centering
		{\includegraphics[scale=0.65]{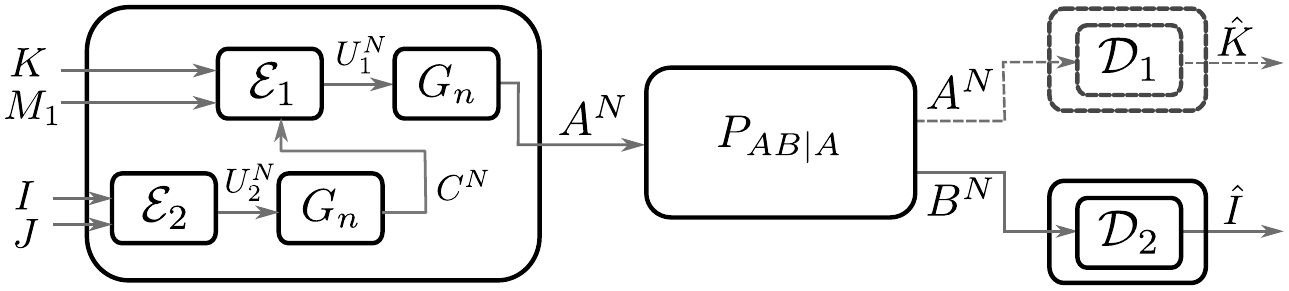}}
		\vspace{-1ex}
		\caption{Block diagram of the superposition polar code.}
		\vspace{-1ex}
		\label{fig:BroadcastChannelPC}
	\end{figure}

	Consider a 2-user \emph{physically} degraded discrete memoryless
        broadcast channel (DM-BC) $P_{AB|A}$ in
        Fig.~\ref{fig:BroadcastChannelPC} where $A$ denotes the channel
        input and $A,B$ denote the output to the first and second receiver,
        respectively. In particular, the channel DMC $P_{B|A}$ is physically
        degraded with respect to the \emph{perfect} channel $P_{A|A}$
        (i.e.,~ $P_{A|A} \succ P_{B|A}$), leading to the Markov chain $A-A-B$.
	We construct the nested polar coding scheme in a similar fashion as
	in \cite{PC_BCGoela2015} as this mimics the nesting of the codebooks
	$\mathscr{C}$ and $\mathscr{A}$ in Step {1)} of the random coding
	construction in~Section~\ref{sec:RandomCoding}. Here, the second
	(weaker) user is able to recover its intended message $I$, while the
	first (stronger)  user is able to recover both messages $K$ and
	$I$. Let $C$ be the auxiliary random variable (cloud center)
	required for  superposition coding over the DM-BC leading to the
	Markov chain $C-A-(A,B)$. As a result, the channel $P_{B|C}$ is also degraded with respect to $P_{A|C}$ (i.e.,~$P_{A|C} \succ P_{B|C}$) \cite[Lemma 3]{PC_BCGoela2015}.  Let $V$ be a matrix of the selected codewords $A^N$ and $C^N$ as
	\begin{equation}\label{eq:codewordsMtrx}
	V\triangleq \begin{bmatrix}
	A^{N}  \\
	C^{N}          
	\end{bmatrix}. 
	\end{equation}
	Now, apply the polar linear transformation $G_n$ as 
	\begin{equation}\label{eq:PolztnMtrx}
	U \triangleq \begin{bmatrix}
	U_{1}^{N}  \\
	U_{2}^{N}          
	\end{bmatrix} = V G_{n}. 
	\end{equation}
	
	First, consider $C^{N} \triangleq U_{2}^{N} G_{n}$ from \eqref{eq:codewordsMtrx} and \eqref{eq:PolztnMtrx} where $U_{2}^{N}$ is generated by the second encoder ${\cal E}_2$ in Fig.~\ref{fig:BroadcastChannelPC}. For $\beta<\frac{1}{2}$ and  ${\delta_N\triangleq 2^{-N^\beta}}$ we define the very high and high entropy sets
	\begin{equation}\label{eq:C_Sets}
	\hspace*{-0.75em}\vspace*{-1ex}\begin{aligned}	
	{\cal V}_{C} &\triangleq \{i\in \llbracket 1,N \rrbracket : H(U_{2}^i|U_{2}^{1:i-1})\!>\!1-\delta_{N}\},\\
	{\cal V}_{C|X} &\triangleq \{i\in \llbracket 1,N \rrbracket : H(U_{2}^i|U_{2}^{1:i-1}X^N)\!>\!1-\delta_{N}\},\\
	{\cal V}_{C|XY} &\triangleq \{i\in \llbracket 1,N \rrbracket : H(U_{2}^i|U_{2}^{1:i-1}X^NY^N)\!>\!1-\delta_{N}\},\\
	{\cal H}_{C|B} &\triangleq \{i\in \llbracket 1,N \rrbracket : H(U_{2}^i|U_{2}^{1:i-1}B^N)\!>\!\delta_{N}\},\\
	{\cal H}_{C|A} &\triangleq \{i\in \llbracket 1,N \rrbracket : H(U_{2}^i|U_{2}^{1:i-1}A^N)\!>\!\delta_{N}\},
	\end{aligned}
	\end{equation}
	which by \cite[Lemma 7]{PC_BCRameBloch16} satisfy
	\vspace*{-1ex}\begin{align*}
	\lim_{N\rightarrow \infty} \frac{|{\cal V}_{C}|}{N} &= H(C),\; & \lim_{N\rightarrow \infty} \frac{|{\cal V}_{C|X}|}{N} &= H(C|X),\\
	\lim_{N\rightarrow \infty} \frac{|{\cal V}_{C|XY}|}{N} &=H(C|XY),\; &\lim_{N\rightarrow \infty} \frac{|{\cal H}_{C|B}|}{N} &=H(C|B),\\
	\lim_{N\rightarrow \infty} \frac{|{\cal H}_{C|A}|}{N} &=H(C|A).
	\end{align*}
	
	These sets are illustrated in Fig.~\ref{fig:CoordiationSets1}. Note
	that the set ${\cal H}_{C|B}$ indicates the noisy bits of the DMC
	$P_{B|C}$ (i.e., the unrecoverable bits of the codeword $C^N$ intended
	for the weaker user in the DM-BC setup in Fig.~3) and is in general not aligned with other sets. Let
	\begin{align*}
	\vspace*{-1ex}{\cal L}_1& \triangleq {\cal V}_{C}\setminus {\cal H}_{C|A} , &&{\cal L}_2\triangleq{\cal V}_{C}\setminus {\cal H}_{C|B},\vspace*{-0.5ex}
	\end{align*}
	
	where the set ${\cal H}_{C|A}$ indicates the noisy bits of the DMC
	$P_{A|C}$ (i.e., the unrecoverable bits of the codeword $C^N$  intended
	for the stronger user). From the relation $P_{A|C} \succ P_{B|C}$ we
	obtain  ${\cal H}_{C|B}^{c} \subseteq {\cal H}_{C|A}^{c} $. This
	ensures that the polarization indices are guaranteed to be aligned
	(i.e.,~ $ {\cal L}_2 \subseteq {\cal L}_1$)
	\cite{PC_BCMondelli2014},\cite[Lemma 4]{PC_BCGoela2015}. As a
	consequence, the bits decodable by the weaker user are also decodable by the stronger user. 
	
	Now, consider $A^{N}\triangleq U_{1}^{N} G_n $ (see
	\eqref{eq:codewordsMtrx} and \eqref{eq:PolztnMtrx}), where
	$U_{1}^{N}$ is generated by the first encoder ${\cal E}_1$ with
	$C^N$ as a side information as seen in
	Fig.~\ref{fig:BroadcastChannelPC}. We define the very high entropy
	sets illustrated in Fig.~\ref{fig:CoordiationSets2} as
	\vspace*{-4ex}
	\begin{figure}[h]
			\centering
			{\includegraphics[scale=0.635]{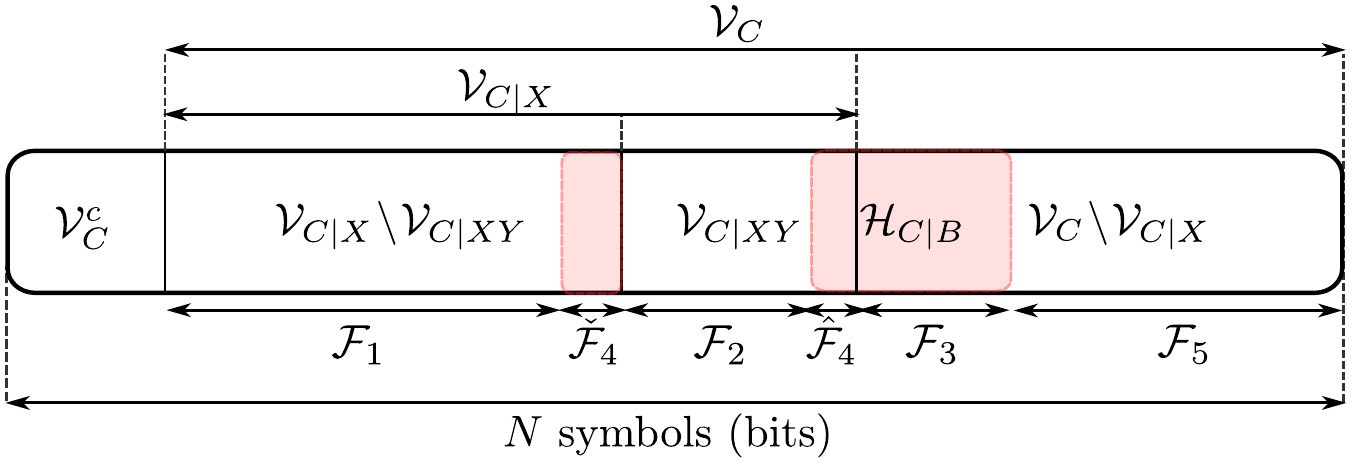}}
			\vspace*{-3.5ex}
			\caption{Index sets for codeword $C$.}
			\vspace*{-4ex}
			\label{fig:CoordiationSets1}
		\end{figure}
		\begin{figure}[h!]
		\centering
		{\includegraphics[scale=0.635]{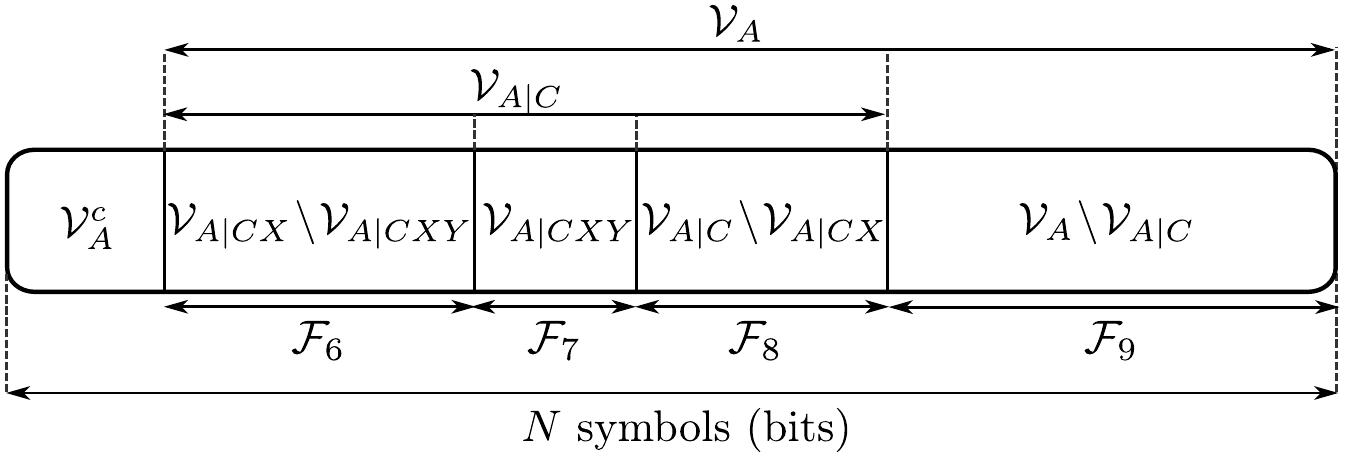}}
		\vspace*{-4ex}
		\caption{Index sets for codeword $A$.}
%		\vspace*{-4ex}
		\label{fig:CoordiationSets2}
	\end{figure}

	\vspace{-0.5em}\begin{equation}\label{eq:A_Sets}
	\begin{aligned}
	{\cal V}_{A} &\triangleq \{i\in \llbracket 1,N \rrbracket : H(U_{1}^i|U_{1}^{1:i-1})\!>\! 1\!-\!\delta_N\!\},\\
	{\cal V}_{A|C} &\triangleq \{i\in \llbracket 1,N \rrbracket : H(U_{1}^i|U_{1}^{1:i-1}C^N)\!>\! 1\!-\!\delta_N \!\},\\
	{\cal V}_{A|CX} &\triangleq \{i\in \llbracket 1,N \rrbracket : H(U_{1}^i|U_{1}^{1:i-1}C^NX^N)\!>\! 1\!-\!\delta_N\!\},\\
	{\cal V}_{A|CXY} &\triangleq \{i\in \llbracket 1,N \rrbracket : H(U_{1}^i|U_{1}^{1:i-1}C^NX^NY^N)\!>\! 1\!-\!\delta_N\!\}
	\end{aligned}
	\end{equation}
	\noindent	
	satisfying
	\vspace{-0.5em}\begin{align*}
	\lim_{N\rightarrow \infty} \frac{|{\cal V}_{A}|}{N} &= H(A),& \lim_{N\rightarrow \infty} \frac{|{\cal V}_{A|CX}|}{N} &=H(A|CX),\\
	\lim_{N\rightarrow \infty} \frac{|{\cal V}_{A|C}|}{N} &= H(A|C), &\lim_{N\rightarrow \infty} \frac{|{\cal V}_{A|CXY}|}{N} &=H(A|CXY).
	\end{align*}
		
	Note that, in contrast to Fig.~\ref{fig:CoordiationSets1}, here
	there is no channel dependent set overlapping with all other sets as
	$P_{A|A}$ is a noiseless  channel with rate $H(A)$ and hence ${\cal H}_{A|A}=\emptyset$.
	
	Accordingly, in terms of the polarization sets in \eqref{eq:C_Sets} and \eqref{eq:A_Sets} we define the sets combining channel resolvability for strong coordination and broadcast channel construction 
	\begin{align*}
	{\cal F}_1&\triangleq ({\cal V}_{C|X}\setminus {\cal V}_{C|XY})\cap {\cal H}_{C|B}^{c}, \\
	{\cal F}_2 &\triangleq {\cal V}_{C|XY} \cap {\cal H}_{C|B}^{c},\\
	{\cal F}_3 &\triangleq  {\cal V}_{C|X}^c \cap {\cal H}_{C|B} ={\cal H}_{C|B}\setminus {\cal H}_{C|BX},\\
	{\cal F}_4&\triangleq {\cal V}_{C|X} \cap {\cal H}_{C|B} = {\cal H}_{C|BX},\\
	\hat{\cal F}_4 &\triangleq {\cal H}_{C|BXY},\\
	\check{\cal F}_4 &\triangleq {\cal H}_{C|BX}\!\setminus\!{\cal H}_{C|BXY},\\
	{\cal F}_5&\triangleq ({\cal V}_{C}\setminus {\cal V}_{C|X})\cap {\cal H}_{C|B}^{c},\\
	{\cal F}_6&\triangleq {\cal V}_{A|CX}\setminus {\cal V}_{A|CXY},\\
	{\cal F}_7&\triangleq {\cal V}_{A|CXY},\\
	{\cal F}_8&\triangleq {\cal V}_{A|C}\setminus {\cal V}_{A|CX},\\
	{\cal F}_9&\triangleq {\cal V}_{A}\setminus {\cal V}_{A|C}.
	\end{align*}

	Finally, with $Y^{N}\triangleq T^{N} G_n $, we define the very high entropy set:
	\begin{equation}\label{eq:Y_Set}
	{\cal V}_{Y|BC} \triangleq \{i\in \llbracket 1,N \rrbracket: H(T^i|T^{1:i-1}B^NC^N)\!>\! \log|\mathcal{Y}|-\delta_N\},
	\end{equation}
	satisfying
	\vspace{-0.5ex}\begin{equation*} 
	\lim_{N\rightarrow \infty} \frac{|{\cal V}_{Y|BC}|}{N} = H(Y|BC). 
	\end{equation*}
	
	This set is useful for expressing the randomized generation of $Y^N$ via simulating the channel $P_{Y|BC}$ in Fig.~\ref{fig:StrongCoordinationAllied} as a source polarization operation \cite{CBJ16}. We now proceed to describe the encoding and decoding algorithms. 
	
	\subsubsection{Encoding}
	The encoding protocol described in Algorithm 1 is performed over $k
	\in \mathbb{N}$ blocks of length $N$. Since for strong coordination
	the goal is to approximate a target joint distribution with a minimum amount of randomness, the encoding scheme performs channel resolvability while reusing a fraction of the common randomness over several blocks (i.e.,  
	randomness recycling) as in \cite{CBJ16}. However, since the communication is over a noisy channel, the encoding scheme also considers a block chaining construction to mitigate the channel noise influence as in \cite{PC_BCMondelli2014,PC_BCRameBloch16,PCempirical2016,GKLJB16}.
	
	\begin{figure}[t]
		\centering
		{\includegraphics[scale=0.49]{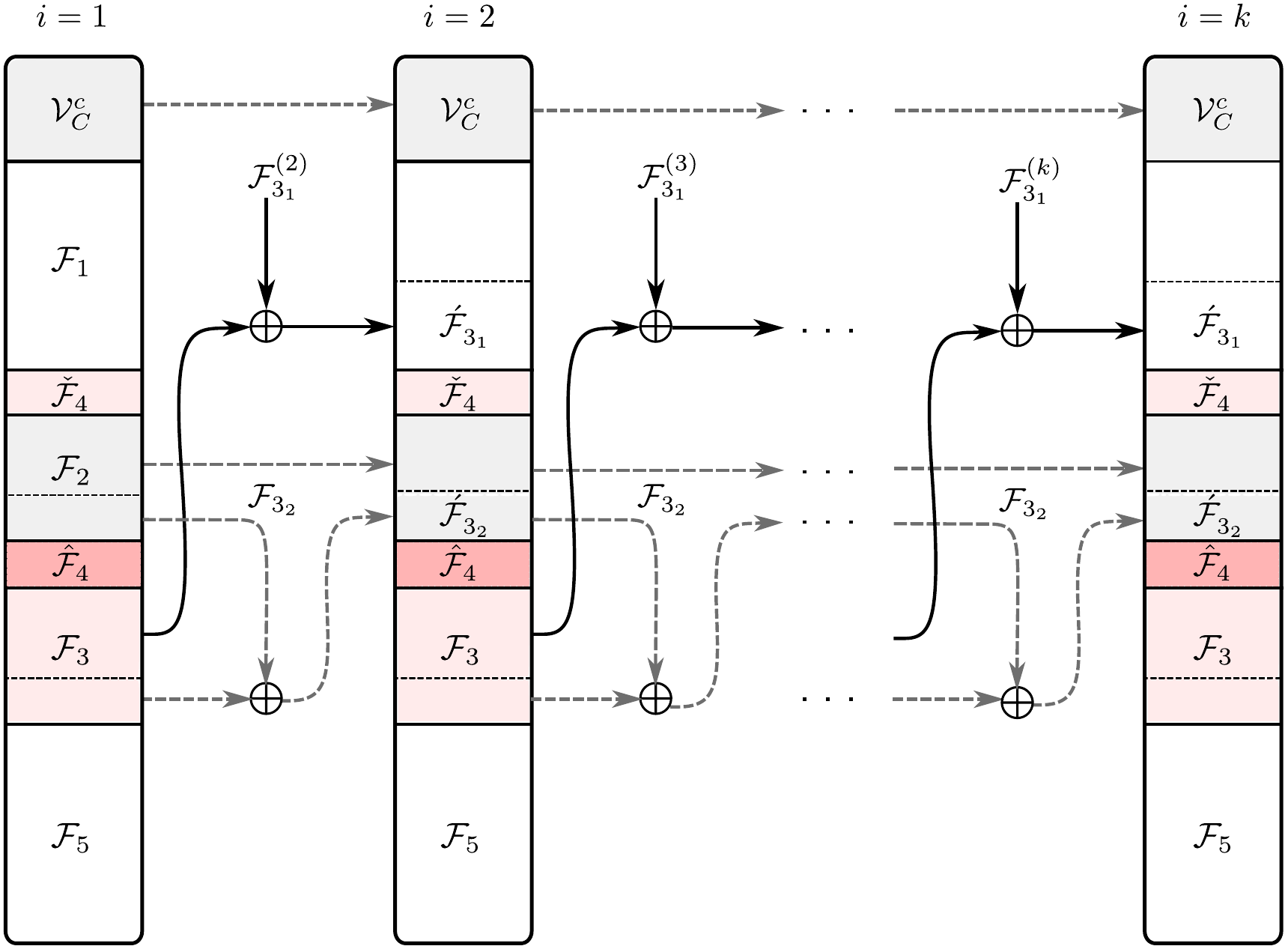}}
		\caption{Chaining construction for block encoding.}
		\vspace*{-2ex}
		\label{fig:ChanningConstruction}
	\end{figure}

	More precisely, as demonstrated in Fig.~\ref{fig:StrongCoordinationAllied}, we are interested
	in successfully recovering the message ${I}$ that is intended for the weak user channel given by
	$P_{B|A}$ in Fig.~\ref{fig:BroadcastChannelPC}. However, the challenge is to communicate the set ${\cal F}_3$ that includes bits of the message $I$ that are corrupted by the channel
	noise. This suggests that we apply a variation of block chaining
        only at  encoder ${\cal E}_2$ generating the codeword $C^N$ as
        follows (see  Fig.~\ref{fig:ChanningConstruction}).
	At encoder ${\cal E}_2$, the set ${\cal F}_3$
	of block $i\in \llbracket 1,k \rrbracket$ is embedded in the reliably decodable bits of ${\cal F}_1
	\cup {\cal F}_2$ of the following block $i+1$. This is possible by
	following the  decodability constraint (see \eqref{equ:JSRR4},
	\eqref{equ:JSRR5} of Theorem 1) that ensures that the size of the set
	${\cal F}_3$ is smaller than the combined size of the sets ${\cal F}_1$
	and ${\cal F}_2$ \cite{PCempirical2016}. However, since
	these sets originally contain uniformly distributed
	common randomness  $J$ \cite{CBJ16}, the bits of ${\cal F}_3$ can be
        embedded while maintaining the uniformity of the randomness  by taking advantage of the Crypto Lemma \cite[Lemma
	2]{Forney2004role}. Then, to ensure that ${\cal F}_3$ is equally distributed over ${\cal F}_1 \cup {\cal F}_2$, ${\cal
		F}_3$ is partitioned according to the ratio between $|{\cal F}_1|$ and $|{\cal
		F}_2|$. To utilize the Crypto Lemma, we introduce $\mathcal{F}_{3_2}$ and
	$\mathcal{F}_{3_1}^{(i)}$, which represent uniformly distributed common 
	randomness used to randomize the information bits of
	$\mathcal{F}_3$. The difference is that $\mathcal{F}_{3_2},$ as ${\cal F}_2,$ represents a fraction of common randomness that can
	be reused over $k$ blocks whereas a realization of the randomness
	in $\mathcal{F}_{3_{1}}^{(i)}$ needs to be
	provided in each new block. Note that, as visualized in Fig.~6, both
        the
        subsets $\acute{\cal F}_{3_1}\subset {\cal F}_1$ and $\acute{\cal
          F}_{3_2} \subset {\cal F}_2$ represent the resulting uniformly distributed bits of ${\cal F}_{3}$ of the previous block, where $|\acute{\cal F}_{3_1}|=|{\cal F}_{3_1}^{}|$ and $|\acute{\cal F}_{3_2}|=|{\cal F}_{3_2}|.$
	Finally, in an additional block $k+1$ we use a good channel  code to reliably transmit the set ${\cal F}_3$ of the last block $k$.

	%	 On the other hand, encoder ${\cal E}_1$, that appends the message intended for the strong %user, performs straight forward polar channel coding.    
	
	\begin{table}[h!]
		\normalsize
		\begin{tabular}{p{0.95\linewidth}}
			\specialrule{.1em}{.05em}{.05em} 
			\rule{0pt}{2.5ex} 
			\hspace{-0.5em}\noindent\textbf{Algorithm 1:} Encoding algorithm at Node $\mathsf X$ for strong coordination\\
			\specialrule{.1em}{.05em}{.05em} \specialrule{.1em}{.05em}{.05em} 
			\rule{0pt}{2.5ex}  
			\hspace{-0.35em}\textbf{Input:} $X_{1:k}^N$, uniformly distributed local randomness bits $M_{1_{1:k}}$ of the size $k|{\cal F}_6|$, common randomness bits $\bar{J}=(\bar{J}_1,\bar{J}_2)$ of sizes $|{\cal F}_2\cup \hat{\cal F}_4|$, and $|{\cal F}_7|$, respectively, and $J_{1:k}$ of size $k|\check{\cal F}_4\cup {\cal F}_1|$ shared with Node $\mathsf Y$.
			
			\textbf{Output:} $\widetilde{A}_{1:k}^N$\\
			1.~\textbf{for} $i=2,\dots,k$ \textbf{do}\\
			2.~${\cal E}_2$ in Fig.~\ref{fig:BroadcastChannelPC} constructs $\widetilde{U}_{2_i}^N$ bit-by-bit as follows:\\
			~~\textbf{if} $i=1$ \textbf{then} 
			\begin{itemize}
				\item $\widetilde{U}_{2_i}^N[{\cal F}_1 \cup \check{\cal F}_4] \leftarrow {J}_{i}$ 
				\item $\widetilde{U}_{2_i}^N[{\cal F}_2 \cup \hat{\cal F}_4 ] \leftarrow \bar{J}_1$
			\end{itemize}
			~~\textbf{else}
			\begin{itemize}
				\item Let ${\cal F}_{3_1}^{(i)}$, ${\cal F}_{3_2}$ be sets of the size $(|{\cal F}_m|\times|{\cal F}_3|)/(|{\cal F}_1|+ |{\cal F}_2|)$ for $m\in\{ 1,2\}.$
				\item $\big(\widetilde{U}_{2_i}^N[({\cal F}_1\setminus\acute{\cal F}_{3_1}) \cup \check{\cal F}_4], {\cal F}_{3_1}^{(i)}\big) \leftarrow {J}_{i}$ 
				\item $\big(\widetilde{U}_{2_i}^N[({\cal F}_2\setminus\acute{\cal F}_{3_2}) \cup \hat{\cal F}_4 ],{\cal F}_{3_2} \big) \leftarrow \bar{J}_1$
				\item $\widetilde{U}_{2_i}^N[\acute{\cal F}_{3_1} ] \leftarrow \widetilde{U}_{2_{i-1}}^N[{\cal F}_3\setminus {\cal F}_{3_2}] \oplus {\cal F}_{3_1}^{(i)} $
				\item $\widetilde{U}_{2_i}^N[\acute{\cal F}_{3_2}] \leftarrow \widetilde{U}_{2_{i-1}}^N[{\cal F}_3\setminus {\cal F}_{3_1}] \oplus {\cal F}_{3_2} $
			\end{itemize}
			~~\textbf{end}
			\begin{itemize}
			\item Given $X_i^N$, successively draw the remaining components of $\widetilde{U}_{2_i}^N$  according to $\tilde{P}_{U_{2_i}^j|U_{2_i}^{1:j-1}X_i^N}$ defined by
			\begin{equation}\label{eq:MsgEncC}
				\hspace*{-1.7ex} \begin{aligned}
			\tilde{P}_{U_{2_i}^j|U_{2_i}^{1:j-1}X_i^N} \triangleq  \begin{cases}
			{Q}_{U_{2}^j|U_{2}^{1:j-1}} & \!j \in  {\cal V}_{C}^{c}, \\
			{Q}_{U_{2}^j|U_{2}^{1:j-1}X^N} & \!j \in {\cal F}_3 \cup {\cal F}_5.
			\end{cases}
			\end{aligned}
			\end{equation} 
		    \end{itemize}
			3.~$\widetilde{C}_{i}^{N} \leftarrow \widetilde{U}_{2_i}^{N} G_{n} $ \\
			4.~${\cal E}_1$ in Fig.~\ref{fig:BroadcastChannelPC} constructs $\widetilde{U}_{1_i}^N$ bit-by-bit as follows: 
			\begin{itemize}
				\item $\widetilde{U}_{1_i}^N[{\cal F}_6] \leftarrow {M}_{1_i}$
				\item $\widetilde{U}_{1_i}^N[{\cal F}_7] \leftarrow \bar{J}_2$
				\item Given $X_i^N$ and $ \widetilde{C}_{i}^{N}$, successively draw the remaining components of $\widetilde{U}_{1_i}^N$ according to $\tilde{P}_{U_{1_i}^j|U_{1_i}^{1:j-1}C_i^NX_i^N}$ defined by
				\vspace{-1.5ex}\begin{equation}\label{eq:MsgEncA}
				\hspace*{-2.52ex} \begin{aligned}
				\tilde{P}_{U_{1_i}^j|U_{1_i}^{1:j-1}C_i^NX_i^N} \triangleq  \begin{cases}
				{Q}_{U_{1}^j|U_{1}^{1:j-1}} & \!j \in  {\cal V}_{A}^{c}, \\
				{Q}_{U_{1}^j|U_{1}^{1:j-1}C^N} & \!j \in {\cal F}_9, \\
				{Q}_{U_{1}^j|U_{1}^{1:j-1}C^NX^N} & \!j \in {\cal F}_8.
				\end{cases}
				\end{aligned}
				\end{equation}
			\end{itemize}
			5.~$\widetilde{A}_i^{N} \leftarrow \widetilde{U}_{1_i}^{N} G_{n} $\\
			6.~Transmit $\widetilde{A}_{i}^N$\\
			7.~\textbf{end for}\\
			\specialrule{.1em}{.05em}{.05em}
		\end{tabular}
	\end{table}
	\vspace*{-1ex}
	
	\subsubsection{Decoding}
	The decoder is described in Algorithm 2.
	Recall that we are only interested
		in the message $\hat{I}$ intended for the weak user channel given by
		$P_{B|A}$ in Figure~\ref{fig:BroadcastChannelPC}. As a result, we
		only state the decoding protocol at ${\cal D}_2$ that recovers the codeword $\widehat{C}^{N}.$ Note that the
	decoding is done in reverse order after receiving the extra $k+1$
	block containing the bits of set ${\cal F}_3$ of the last block
	$k$. In particular, in each block $i\in[1,k-1]$ the bits in
	$\mathcal{F}_3$ are obtained by successfully recovering the bits in both
	$\mathcal{F}_1$ and $\mathcal{F}_2$ in block $i+1$.
	
	\begin{table}[t!]
		\normalsize
		\begin{tabular}{p{0.95\linewidth}}
			\specialrule{.1em}{.05em}{.05em} 
			\rule{0pt}{2.5ex} 
			\hspace{-0.5em}\noindent\textbf{Algorithm 2:} Decoding algorithm at Node $\mathsf Y$ for strong coordination\\
			\specialrule{.1em}{.05em}{.05em}	\specialrule{.1em}{.05em}{.05em}
			\rule{0pt}{2.5ex}  
			\hspace{-0.35em}\textbf{Input:} $B_{1:k}^N,$ uniformly distributed common randomness, $\bar{J}_1$, and $J_{1:k}$ shared with Node $\mathsf X$.\\ 
		
			\textbf{Output:} $\widetilde{Y}_{1:k}^N$\\
			1.~\textbf{For} block $i=k,\dots,1$ \textbf{do}\\ 
			2.~${\cal D}_2$ in Fig.~\ref{fig:BroadcastChannelPC} constructs $\widehat{U}_{2_i}^N$ bit-by-bit as follows:
			\begin{itemize}
				\item $\big(\widehat{U}_{2_i}^N[({\cal F}_1\setminus\acute{\cal F}_{3_1}^{}) \cup \check{\cal F}_4], {\cal F}_{3_1}^{(i)}\big) \leftarrow {J}_{i}$ 
				\item $\big(\widehat{U}_{2_i}^N[({\cal F}_2\setminus\acute{\cal F}_{3_2}) \cup \hat{\cal F}_4 ],{\cal F}_{3_2}\big) \leftarrow \bar{J}_1$
				\item Given $B_i^N$ successively draw the components of  $\widehat{U}_{2_i}^N$ according to 
				$\tilde{P}_{U_{2_i}^j|U_{2_i}^{j-1},B_i^N}$ defined by
				\begin{equation}\label{eq:MsgDec}
				\hspace{-5ex}\begin{aligned}
				\tilde{P}_{U_{2_i}^j|U_{2_i}^{j-1}B_i^N}\triangleq\begin{cases}
				{Q}_{U_{2}^j|U_{2}^{1:j-1}}\! & \!j\! \in  {\cal V}_{C}^{c}, \\
				{Q}_{U_{2}^j|U_{2}^{1:j-1}B_i^N}\!\! & \!j\! \in \acute{\cal F}_{3_2}\!\cup\!\acute{\cal F}_{3_1}\!\cup\!{\cal F}_5.
				\end{cases}
				\end{aligned}\hspace{-1ex}
				\end{equation}
			\end{itemize}
			3.~\textbf{if} $i=k$ \textbf{then}
			\begin{itemize}
				\item $\widehat{U}_{2_i}^N[{\cal F}_3]\leftarrow B_{k+1}^N $
			\end{itemize}
				~~\textbf{else}
			\begin{itemize}
				\item $ \widehat{U}_{2_{i}}^N[{\cal F}_3\setminus {\cal F}_{3_2}]  \leftarrow \widehat{U}_{2_{i+1}}^N[\acute{\cal F}_{3_1} ]\oplus {\cal F}_{3_1}^{(i+1)} $
				\item $\widehat{U}_{2_{i}}^N[{\cal F}_3\setminus {\cal F}_{3_1}]   \leftarrow  \widehat{U}_{2_{i+1}}^N[\acute{\cal F}_{3_2}]\oplus {\cal F}_{3_2} $
			\end{itemize}
			4.~Let \begin{itemize}
				\item $\widehat{U}_{2_i}^N[\acute{\cal F}_{3_1}^{} ] \leftarrow {\cal F}_{3_1}^{(i)}$ 
				\item $\widehat{U}_{2_i}^N[\acute{\cal F}_{3_2}] \leftarrow {\cal F}_{3_2} $
				\end{itemize}	
			5.~$\widehat{C}_{i}^{N} \leftarrow \widehat{U}_{2_i}^{N} G_{n} $ \\
			6.~Channel simulation: given $\widehat{C}_{i}^N$ and ${B_i}^N,$ successively draw the components of $\widetilde{T}_{i}^{N}$ according to 
			\begin{equation}\label{eq:seqGen}
			\hspace{-1.16ex}\begin{aligned}
			\tilde{P}_{T_{i}^j|T_{i}^{1:j-1}B_i^NC_i^N} \triangleq  \begin{cases}
			1/|{\cal Y}| & \!j \in {\cal V}_{Y|BC}, \\
			{Q}_{T_{}^j|T_{}^{1:j-1}B^NC^N} & \!j \in {\cal V}_{Y|BC}^{c}.
			\end{cases}
			\end{aligned}
			\end{equation}
			5.~$\widetilde{Y}_i^N \leftarrow \widetilde{T}_i^N G_n$\\
			6.~\textbf{end for}\\
			\specialrule{.1em}{.05em}{.05em} 
		\end{tabular}
	\vspace*{-1ex}
	\end{table}
	\vspace*{-1ex}

%%%%%%%%%%%%%%%%%%%%%%%%%%%%%%%%%%%%%%%%%%%%%%%%%%%%%%%%%%%%%%%%%%%%%%%%%%%%%%%%%%%%5
	\subsection{Scheme Analysis}
	We now provide an analysis of the coding scheme of Section III. The analysis is based on KL divergence which upper bounds the total variation in \eqref{Thm:JointCRR} by Pinsker's inequality. We start the analysis with a set of sequential lemmas. In particular, Lemma 1 is useful to show in Lemma 2 that the strong coordination scheme based on channel resolvability holds for each block individually regardless of the randomness recycling. 

	\begin{lemma}\label{lemma1}
		For block $i \in \llbracket 1,k \rrbracket ,$ we have
		\vspace{-0.8ex}\begin{equation*}
		\mathbb{D}(Q_{A^NC^NX^N}||\tilde{P}_{A_i^NC_i^NX_i^N}) \leq 2N\delta_{N}.
		\end{equation*} 
	\end{lemma}

	\begin{proof}
		We have
		\hspace{-1em}\begin{equation*}
		\begin{split} 
		\hspace*{-0.9em}&\mathbb{D}({Q}_{A^NC^NX^N}||\tilde{P}_{A_i^NC_i^NX_i^N})\\ 
		&\stackrel{(a)}{=} \mathbb{D}({Q}_{U_{1}^NU_{2}^NX^N}||\tilde{P}_{U_{1_i}^NU_{2_i}^NX_i^N}) \\
		&\stackrel{(b)}{=} \mathbb E_{Q_{X^N}}\Big[ \mathbb{D}({Q}_{U_{1}^NU_{2}^N|X^N}||\tilde{P}_{U_{1_i}^NU_{2_i}^N|X_i^N}) \Big] \\
	 	&\stackrel{}{=} \mathbb E_{Q_{X^N}}\Big[ \mathbb{D} ({Q}_{U_{2}^N|X^N}{Q}_{U_{1}^N|U_{2}^NX^N}||\tilde{P}_{U_{2_i}^N|X_i^N}\tilde{P}_{U_{1_i}^N|U_{2_i}^NX_i^N}) \Big] \\
		&\stackrel{(c)}{=} \mathbb E_{Q_{X^N}}\Big[ \mathbb{D}({Q}_{U_{2}^N|X^N}||\tilde{P}_{U_{2_i}^N|X_i^N})\\ 		  
		&\qquad\quad\quad \;\;+\mathbb{D}({Q}_{U_{1}^N|U_{2}^NX^N}||\tilde{P}_{U_{1_i}^N|U_{2_i}^NX_i^N}) \Big] \\
		&\stackrel{(d)}{=} \sum_{j=1}^{N} \mathbb E_{Q_{U_{2}^{1:j-1}X^N}}\Big[ \mathbb{D} ({Q}_{U_{2}^j|U_{2}^{1:j-1}X^N}||\tilde{P}_{U_{2_i}^j|U_{2_i}^{1:j-1}X_i^N}) \Big]\\
		&+\! \sum_{j=1}^{N} \mathbb E_{Q_{U_{1}^{1:j-1}\!U_{2}^{N}\!X^N}}\!\Big[ \mathbb{D} ({Q}_{U_{1}^j\!|U_{1}^{1:j-1}\!U_{2}^N\!X^N}\!||\tilde{P}_{U_{1_i}^j\!|U_{1_i}^{1:j-1}\!U_{2_i}^N\!X_i^N}) \Big] 
	\end{split}
                \end{equation*}
                \vspace*{-0.9em}
                \begin{equation*}
		\begin{split} 	
		\hspace*{-0.9em}
		&\stackrel{(e)}{=}\!\!\sum_{j \notin {\cal F}_3 \cup {\cal
                    F}_5} \mathbb E_{Q_{U_{2}^{1:j-1}X^N}}\Big[ \mathbb{D}
                ({Q}_{U_{2}^j|U_{2}^{1:j-1}X^N}||\tilde{P}_{U_{2_i}^j|U_{2_i}^{1:j-1}X_i^N})
                \Big] \\
               	&\!+\!\!\!\sum_{j \notin {\cal F}_8}\!\!\mathbb E_{Q_{U_{1}^{1:j-1}\!U_{2}^{N}\!X^N}}\!\Big[ \mathbb{D} ({Q}_{U_{1}^j|U_{1}^{1:j-1}\!U_{2}^N\!X^N}||\tilde{P}_{U_{1_i}^j|U_{1_i}^{1:j-1}\!U_{2_i}^N\!X_i^N}) \Big]\\
		&\stackrel{(f)}{=} \!\!\sum_{j \in {\cal V}_{C}^{c}\cup{\cal V}_{C|X}} \!\!\!\!\mathbb E_{Q_{U_{2}^{1:j-1}X^N}} \Big[ \mathbb{D} ({Q}_{U_{2}^j|U_{2}^{1:j-1}X^N}||\tilde{P}_{U_{2_i}^j|U_{2_i}^{1:j-1}X_i^N}) \Big] \\
		&+\sum_{j \in {\cal V}_{A}^{c}\cup {\cal V}_{A|CX} \cup {\cal V}_{A}\!\setminus\!{\cal V}_{A|C} } \!\!\!\!\mathbb E_{Q_{U_{1}^{1:j-1}U_{2}^{N}X^N}} \!\Big[ \\
		&\qquad \qquad \qquad \qquad \mathbb{D} ({Q}_{U_{1}^j|U_{1}^{1:j-1}U_{2}^NX^N}||\tilde{P}_{U_{1_i}^j|U_{1_i}^{1:j-1}U_{2_i}^NX_i^N}) \Big]\\
		&\stackrel{(g)}{=} \sum_{j \in {\cal V}_{C}^{c}} \Big( H(U_{2}^j|U_{2}^{1:j-1}) - H(U_{2}^j|U_{2}^{1:j-1}X^N) \Big)\\ 
		&+\sum_{j \in {\cal V}_{C|X}} \Big( 1-H(U_{2}^j|U_{2}^{1:j-1}X^N) \Big)\\
		&+\sum_{j \in {\cal V}_{A}^{c}} \Big( H(U_{1}^j|U_{1}^{1:j-1}) -H(U_{1}^j|U_{1}^{1:j-1}U_2^NX^N) \Big)\\
		&+\sum_{j \in {\cal V}_{A|CX}} \!\Big( 1- H(U_{1}^j|U_{1}^{1:j-1}U_2^NX^N) \Big)\\
		&+ \sum_{j \in {\cal V}_{A|C}^{c} \!\setminus\!{\cal V}_{A}^{c}} \Big( H(U_{1}^j|U_{1}^{1:j-1}U_{2}^N) - H(U_{1}^j|U_{1}^{1:j-1}U_{2}^NX^N) \Big)\\
		&\stackrel{(h)}{=} \sum_{j \in {\cal V}_{C}^{c}} \Big( H(U_{2}^j|U_{2}^{1:j-1}) - H(U_{2}^j|U_{2}^{1:j-1}X^N) \Big)\\ 
		&+\sum_{j \in {\cal V}_{C|X}} \Big( 1-H(U_{2}^j|U_{2}^{1:j-1}X^N) \Big)\\
		&+\sum_{j \in {\cal V}_{A}^{c}} \Big( H(U_{1}^j|U_{1}^{1:j-1}) -H(U_{1}^j|U_{1}^{1:j-1}C^NX^N) \Big)\\
		&+\sum_{j \in {\cal V}_{A|CX}} \!\Big( 1- H(U_{1}^j|U_{1}^{1:j-1}C^NX^N) \Big)\\
		&+\!\!\sum_{j \in {\cal V}_{A|C}^{c} \!\setminus\!{\cal V}_{A}^{c}} \Big( H(U_{1}^j|U_{1}^{1:j-1}C^N) - H(U_{1}^j|U_{1}^{1:j-1}C^NX^N) \Big)\\
		&\stackrel{(i)}{\leq} (|{\cal V}_{C}^{c}|+|{\cal V}_{C|X}|+ |{\cal V}_{A|XC}|+|{\cal V}_{A|C}^{c}|)\delta_N \leq 2N\delta_N
		 \end{split}
		 \end{equation*}
\noindent where
\begin{itemize} 
	\item[(a)] holds by invertibility of $G_n$;
	\item[(b)]\hspace*{-0.5em}\;-\;(d)~follows from the chain rule of the KL divergence \cite{EoIT:2006};
	\item[(e)] results from the definitions of the conditional distributions in \eqref{eq:MsgEncC}, and \eqref{eq:MsgEncA};
	\item[(f)] follows from the definitions of the index sets as shown in Figures~\ref{fig:CoordiationSets1} and~\ref{fig:CoordiationSets2};
	\item[(g)] results from the encoding of $\widetilde{U}_{1_i}^N$ and $\widetilde{U}_{2_i}^N$ bit-by-bit at ${\cal E}_1$ and ${\cal E}_2,$ respectively, with uniformly distributed randomness bits and message bits. These bits are generated by applying successive cancellation encoding using previous bits and side information with conditional distributions defined in \eqref{eq:MsgEncC} and \eqref{eq:MsgEncA};
	\item[(h)] holds by the one-to-one relation between ${U}_{2}^N$ and $C^N$;
	\item[(i)] follows from the sets defined in~\eqref{eq:C_Sets} and~\eqref{eq:A_Sets}.
\end{itemize}
\end{proof}
	\vspace{-3ex}
	\begin{lemma}\label{lemma2}
		For block $i \in \llbracket 1,k \rrbracket ,$ we have
		\vspace*{-1ex}
		\begin{align*}
		\mathbb{D}(&\tilde{P}_{X_i^NY_i^N}||Q_{X^NY^N})\\ &\leq\mathbb{D}(\tilde{P}_{X_i^NA_i^NC_i^NB_i^N\widehat{C}_i^NY_i^N}||Q_{X^NA^NC^NB^N\widehat{C}^NY^N})\leq \delta_{N}^{(2)} 
		\end{align*}
		\vspace*{-1ex}
		where $\delta_{N}^{(2)}\triangleq {\cal O}(\sqrt{N^3 \delta_N}).$ %N^{3/2}\delta_{N}^{1/2})
	\end{lemma}
\begin{proof}
	Consider the argument shown at the top of the following page. In this argument:
	
	\begin{figure*}[!t]
		\begin{equation}
		\begin{split} \notag
		\mathbb{D}(\tilde{P}_{X_i^NA_i^NC_i^NB_i^N\widehat{C}_i^NY_i^N}&||Q_{X^NA^NC^NB^N\widehat{C}^NY^N})\\ 
		&\stackrel{}{=} \mathbb{D}(\tilde{P}_{Y_i^N|X_i^NA_i^NC_i^NB_i^N\widehat{C}_i^N}\tilde{P}_{X_i^NA_i^NC_i^NB_i^N\widehat{C}_i^N}||Q_{Y^N|X^NA^NC^NB^N\widehat{C}^N}Q_{X^NA^NC^NB^N\widehat{C}^N}) \\
		&\stackrel{(a)}{=} \mathbb{D}(\tilde{P}_{Y_i^N|B_i^N\widehat{C}_i^N}\tilde{P}_{X_i^NA_i^NC_i^NB_i^N\widehat{C}_i^N}||Q_{Y^N|B^N\widehat{C}^N}Q_{X^NA^NC^NB^N\widehat{C}^N}) \\
		&\stackrel{}{=} \mathbb{D}(\tilde{P}_{Y_i^N|B_i^N\widehat{C}_i^N}\tilde{P}_{B_i^N\widehat{C}_i^N|X_i^NA_i^NC_i^N}\tilde{P}_{X_i^NA_i^NC_i^N}||Q_{Y^N|B^N\widehat{C}^N}Q_{B^N\widehat{C}^N|X^NA^NC^N}Q_{X^NA^NC^N}) \\
		&\stackrel{(b)}{=} \mathbb{D}(\tilde{P}_{Y_i^N|B_i^N\widehat{C}_i^N}\tilde{P}_{B_i^N\widehat{C}_i^N|A_i^NC_i^N}\tilde{P}_{X_i^NA_i^NC_i^N}||Q_{Y^N|B^N\widehat{C}^N}Q_{B^N\widehat{C}^N|A^NC^N}Q_{X^NA^NC^N}) \\
		&\stackrel{(c)}{\leq} \hat{\delta}_N^{(2)} + \mathbb{D}(\tilde{P}_{Y_i^N|B_i^N\widehat{C}_i^N}\tilde{P}_{B_i^N\widehat{C}_i^N|A_i^NC_i^N}\tilde{P}_{X_i^NA_i^NC_i^N}||\tilde{P}_{Y_i^N|B_i^N\widehat{C}_i^N}\tilde{P}_{B_i^N\widehat{C}_i^N|A_i^NC_i^N}Q_{X^NA^NC^N}) \\
		&\quad + \mathbb{D}(\tilde{P}_{Y_i^N|B_i^N\widehat{C}_i^N}\tilde{P}_{B_i^N\widehat{C}_i^N|A_i^NC_i^N}Q_{X^NA^NC^N}||Q_{Y^N|B^N\widehat{C}^N}Q_{B^N\widehat{C}^N|A^NC^N}Q_{X^NA^NC^N})\\
		&\stackrel{(d)}{=} \hat{\delta}_N^{(2)} + \mathbb{D}(\tilde{P}_{X_i^NA_i^NC_i^N}||Q_{X^NA^NC^N}) 
		+\mathbb{D}(\tilde{P}_{Y_i^N|B_i^N\widehat{C}_i^N}\tilde{P}_{B_i^N\widehat{C}_i^N|A_i^NC_i^N}||Q_{Y^N|B^N\widehat{C}^N}Q_{B^N\widehat{C}^N|A^NC^N})\\
		&\stackrel{(e)}{\leq} \hat{\delta}_N^{(2)} + \delta_N^{(1)}
		+\mathbb{D}(\tilde{P}_{Y_i^N|B_i^N\widehat{C}_i^N}\tilde{P}_{B_i^N\widehat{C}_i^N|A_i^NC_i^N}||Q_{Y^N|B^N\widehat{C}^N}Q_{B^N\widehat{C}^N|A^NC^N})\\
		&\stackrel{(f)}{=} \hat{\delta}_N^{(2)} + \delta_N^{(1)}
		+\mathbb{D}(\tilde{P}_{Y_i^N|B_i^N\widehat{C}_i^N}||Q_{Y^N|B^N\widehat{C}^N})+\mathbb{D}(\tilde{P}_{B_i^N\widehat{C}_i^N|A_i^NC_i^N}||Q_{B^N\widehat{C}^N|A^NC^N})\\
		&\stackrel{(g)}{\leq} \hat{\delta}_N^{(2)} + \delta_N^{(1)}
		-N\log(\mu_{YB\widehat{C}})\sqrt{2\ln 2} \sqrt{\mathbb{D}(Q_{Y^N|B^N\widehat{C}^N}||\tilde{P}_{Y_i^N|B_i^N\widehat{C}_i^N})}\\
		&\quad -N\log(\mu_{ACB\widehat{C}})\sqrt{2\ln 2} \sqrt{\mathbb{D}(Q_{B^N\widehat{C}^N|A^NC^N}||\tilde{P}_{B_i^N\widehat{C}_i^N|A_i^NC_i^N})} \\
		&\stackrel{(h)}{\leq} \hat{\delta}_N^{(2)} + \delta_N^{(1)} -N\log(\mu_{YB\widehat{C}})\sqrt{2\ln 2} \sqrt{N \delta_N} -N\log(\mu_{ACB\widehat{C}})\sqrt{2\ln 2} \sqrt{N \delta_N} 
		\end{split}
		\end{equation}
		\hrulefill
	\end{figure*}

%\vspace*{-1ex}
\begin{itemize} 
	\item[(a)]\hspace*{-0.5em}\;-\;(b) results from the Markov chain $X^N\!-\!A^NC^N\!-\!B^N\widehat{C}^N\!-\!Y^N$; 
	\item[(c)] follows from \cite[Lemma 16]{CBJ16} where 
		\vspace*{-1ex}
		\begin{align*}\hat{\delta}_N^{(2)} &\triangleq -N\log(\mu_{XACB\widehat{C}Y})\sqrt{2\ln 2} \sqrt{2N\delta_{N}},\\
		\mu_{XACB\widehat{C}Y}&\triangleq {\textstyle\min^{*}_{x,y,a,c,b,\hat{c}}} \big(Q_{XACB\widehat{C}Y}\big);
		\end{align*}
		\vspace*{-1ex}
	\item[(d)] follows from the chain rule of KL divergence \cite{EoIT:2006};
	\item[(e)] holds by Lemma~\ref{lemma1} and \cite[Lemma 14]{CBJ16} where 
		\vspace{-1ex}
		\begin{align*}\delta_N^{(1)} &\triangleq -N\log(\mu_{XAC})\sqrt{2\ln 2} \sqrt{2N\delta_{N}},\\
		 \mu_{XAC} &\triangleq {\textstyle\min^{*}_{x,a,c}} \big(Q_{XAC}\big);
		 \end{align*}
	\item[(f)] follows from the chain rule of KL divergence \cite{EoIT:2006};
	\item[(g)] holds by \cite[Lemma 14]{CBJ16}, where
	\vspace{-1ex}
	\begin{align*}\mu_{ACB\widehat{C}} &\triangleq{\textstyle\min^{*}_{a,c,b,\hat{c}}} \big(Q_{ACB\widehat{C}}\big),\\
	\mu_{YB\widehat{C}} &\triangleq{\textstyle\min^{*}_{y,b,\hat{c}}} \big(Q_{YB\widehat{C}}\big);\end{align*}
	\item[(h)] holds by bounding the terms\\ $\mathbb{D}(Q_{B^N\widehat{C}^N|A^NC^N}||\tilde{P}_{B_i^N\widehat{C}_i^N|A_i^NC_i^N}),$ and\\ $\mathbb{D}(Q_{Y^N|B^N\widehat{C}^N}||\tilde{P}_{Y_i^N|B_i^N\widehat{C}_i^N})$, as follows: 
%	\vspace{-1ex}
	\begin{equation*}
	\begin{split} 
	&\mathbb{D}(Q_{B^N\widehat{C}^N|A^NC^N}||\tilde{P}_{B_i^N\widehat{C}_i^N|A_i^NC_i^N})\\
		&\stackrel{(a)}{=} \mathbb{D}(Q_{B^N|A^N}Q_{\widehat{C}^N|B^N}||{Q}_{B^N|A^N}\tilde{P}_{\widehat{C}_i^N|B_i^N})\\ 
		&\stackrel{}{=} \mathbb{D}(Q_{\widehat{C}^N|B^N}||\tilde{P}_{\widehat{C}_i^N|B_i^N})\\ 
		&\stackrel{(b)}{=} \mathbb{D}(Q_{\widehat{U}^N|B^N}||\tilde{P}_{\widehat{U}_i^N|B_i^N})\\
		&\stackrel{(c)}{=}\!\sum_{j=1}^{N} \mathbb E_{Q_{U_{2}^{1:j-1}B^N}}\Big[ \mathbb{D}(Q_{U_{2}^j|U_{2}^{1:j-1}B^N}||\tilde{P}_{U_{2_i}^j|U_{2_i}^{1:j-1}B_i^N}) \Big]\\
		&\stackrel{(d)}{=}\!\sum_{j \in {\cal V}_{C}^{c}}\!\mathbb E_{Q_{U_{2}^{1:j-1}B^N}}\Big[ \mathbb{D}(Q_{U_{2}^j|U_{2}^{1:j-1}B^N}||\tilde{P}_{U_{2_i}^j|U_{2_i}^{1:j-1}B_i^N}\!) \Big] \\
		&+\!\!\,\sum_{\!\!j \in {\cal H}_{C|B}\cup {\cal V}_{C|X}\!\!}\!\!\!\!\!\!\!\!\mathbb E_{Q_{U_{2}^{1:j-1}\!B^N}}\!\Big[ \mathbb{D}(Q_{\!U_{2}^j\!|U_{2}^{1:j-1}\!B^N}\!||\tilde{P}_{\!U_{2_i}^j\!|U_{2_i}^{1:j-1}\!B_i^N}\!) \Big]\\
		&\stackrel{(e)}{=}\!\!\sum_{j \in {\cal V}_{C}^{c}}\!\!\!\Big( H(U_{2}^j|U_{2}^{1:j-1}) - H(U_{2}^j|U_{2}^{1:j-1}B^N) \Big)\\
		&+\sum_{j \in {\cal H}_{C|B}\cup {\cal V}_{C|X}}\!\!\!\Big( 1 - H(U_{2}^j|U_{2}^{1:j-1}B^N) \Big)\\ 
		&\stackrel{(f)}{\leq} |{\cal V}_{C}^{c}|\delta_N +|{\cal H}_{C|B}\cup {\cal V}_{C|X}|\delta_N \leq N\delta_N,  
		 \end{split}
		 \end{equation*}
	\noindent where
	\begin{itemize} 
		\item[(a)] results from the Markov chain $C-A-B-\widehat{C}$ and the fact that $\tilde{P}_{B_i^N|A_i^N}=Q_{B^N|A^N}$;
		\item[(b)] holds by the one-to-one relation between ${U}_{2}^N$ and $C^N$;
		\item[(c)] follows from the chain rule of KL divergence \cite{EoIT:2006};
		\item[(d)]\hspace*{-0.5em}\;-\;(e) results from the definitions of the conditional distributions in~\eqref{eq:MsgDec};
		\item[(f)] follows from the sets defined in~\eqref{eq:C_Sets}. 
	\end{itemize}
	\begin{equation*}
	\hspace*{-1.5em}\begin{split} 
	&\mathbb{D}(Q_{Y^N|B^N\widehat{C}^N}||\tilde{P}_{Y_i^N|B_i^N\widehat{C}_i^N})\\ 
	&\!\stackrel{(a)}{=}\!\sum_{j=1}^{N} \mathbb E_{Q_{T^{1:j-1}\!B_i^N\!\widehat{C}_i^N}}\!\Big[ \mathbb{D}(Q_{T^j|T^{1:j-1}\!B^N\!\widehat{C}^N}||\tilde{P}_{T^j|T^{1:j-1}\!B_i^N\!\widehat{C}_i^N}) \Big] \\
	&\!\stackrel{(b)}{=}\!\!\!\!\!\!\sum_{j \in {\cal V}_{Y|BC}}\!\!\!\!\!\mathbb E_{Q_{T^{1:j-1}\!B_i^N\!\widehat{C}_i^N}}\!\Big[ \mathbb{D}(Q_{T^j|T^{1:j-1}\!B^N\!\widehat{C}^N}||\tilde{P}_{T^j|T^{1:j-1}\!B_i^N\!\widehat{C}_i^N}) \Big] \\
	&\!\stackrel{(c)}{=}\!\!\sum_{j \in {\cal V}_{Y|BC}}\!\!\!\Big( \log|{\cal Y}| - H(T^j|T^{1:j-1}B^NC^N) \Big)\\ 
	&\!\stackrel{(d)}{\leq} |{\cal V}_{Y|BC}|\delta_N \leq N\delta_N, 
	\end{split}
	\end{equation*}
	
	\noindent where
		\begin{itemize} %results from the Markov chain $C-A-B-\widehat{C}$ 
			\item[(a)] follows from the chain rule of KL divergence \cite{EoIT:2006};
			\item[(b)]\hspace*{-0.5em}\;-\;(c) results from the definitions of the conditional distribution in~\eqref{eq:seqGen};
			\item[(d)] follows from the set defined in~\eqref{eq:Y_Set}. 
		\end{itemize}
\end{itemize}
\end{proof}

Now, Lemmas 3 and 4 provide the independence between two consecutive blocks and the independence between all blocks based on the results of Lemma 2. 

\begin{lemma}\label{lemma3}
	For block $i \in \llbracket 2,k \rrbracket ,$ we have
	$$ \mathbb{D}(\tilde{P}_{X_{i-1:i}^NY_{i-1:i}^N\bar{J}_{1}}||\tilde{P}_{X_{i-1}^NY_{i-1}^N\bar{J}_{1}}\tilde{P}_{X_{i}^NY_{i}^N}) \leq \delta_{N}^{(3)} $$
	where $\delta_{N}^{(3)}\triangleq {\cal O}(\sqrt[4]{N^{15}\delta_N}).$  %N^{15/4}\delta_{N}^{1/4}
\end{lemma}
\begin{proof} We reuse the proof of \cite[Lemma 3]{CBJ16} with substitutions $q_{U^{1:N}} \leftarrow Q_{C^N},$ $q_{Y^{1:N}} \leftarrow Q_{X^N\!Y^N},$ $\tilde{p}_{U_i^{1:N}} \leftarrow \tilde{P}_{C_i^N},$ ${\tilde{p}_{Y_i^{1:N}}\!\leftarrow\!\tilde{P}_{Y_i^N\!X_i^N}},$ and $\bar{R}_1 \leftarrow \bar{J}_{1}$. This will result in the Markov chain $\!X_{i-1}^N\!\widetilde{Y}_{i-1}^N-\bar{J}_{1}-X_i^N\!\widetilde{Y}_i^N\!$ replacing the chain in~\cite[Lemma~3]{CBJ16}. 
\end{proof}

\begin{lemma}\label{lemma4}
	We have
	$$ \mathbb{D}\Big(\tilde{P}_{X_{1:k}^NY_{1:k}^N}||\prod_{i=1}^{k}\tilde{P}_{X_{i}^NY_{i}^N}\Big) \leq (k-1)\delta_{N}^{(3)} $$
	where $\delta_{N}^{(3)}$ is defined in Lemma 3. 
\end{lemma}
\begin{proof} We reuse the proof of \cite[Lemma~4]{CBJ16} with substitutions $\tilde{p}_{Y_i^{1:N}} \leftarrow \tilde{P}_{X_i^NY_i^N},$ and $\bar{R}_1 \leftarrow \bar{J}_{1}$. This will result in the Markov chain $X_{1:i-2}^N\widetilde{Y}_{1:i-2}^N-\bar{J}_{1}X_{i-1}^N\widetilde{Y}_{i-1}^N-X_i^N\widetilde{Y}_i^N$ replacing the chain in \cite[Lemma~4]{CBJ16}.
\end{proof}

Finally, by the results of Lemma 4 we can show in Lemma 5 that the target distribution $Q_{X^NY^N}$ is approximated asymptotically over all blocks jointly. 
\begin{lemma}\label{lemma5}
	We have
	$$ \mathbb{D}\Big(\tilde{P}_{X_{1:k}^NY_{1:k}^N}||Q_{X^{1:kN}Y^{1:kN}}\Big) \leq \delta_{N}^{(4)}.$$
	where $\delta_{N}^{(4)}\triangleq {\cal O}(k^{3/2}N^{23/8}\delta_{N}^{1/8})$ 
\end{lemma}
\begin{proof} We reuse the proof of \cite[Lemma 5]{CBJ16} with substitutions 
	 $q_{Y^{1:N}} \leftarrow Q_{X^NY^N},$  $\tilde{p}_{Y_i^{1:N}} \leftarrow \tilde{P}_{Y_i^NX_i^N}.$ 
\end{proof}

\begin{theorem}\label{Thm:PCRR} The polar coding scheme described in Algorithms 1, 2 achieves the region stated in Theorem~\ref{Thm:JointCRR}. It satisfies \eqref{Thm:JointCRR} for a binary input DMC channel and a target distribution $q_{XY}$ defined over ${\cal X}\times {\cal Y}$, with an axillary random variable $C$ defined over the binary alphabet. 
\end{theorem}
\begin{proof}
	The common randomness rate $R_o$ is given as
	\begin{align} \label{eq:Ro}
	\frac{|\bar{J}_1|+|J_{1:k}|}{kN}
	&= \frac{|{\cal V}_{C|XY}|+k|{\cal V}_{C|X}\setminus {\cal V}_{C|XY}|}{Nk} \notag\\
	&= \frac{|{\cal V}_{C|XY}|}{kN} +\frac{|{\cal V}_{C|X}\setminus {\cal V}_{C|XY}|}{N} \notag\\
	&\xrightarrow{N\rightarrow\infty} \frac{H(C|XY)}{k}+ I(Y;C|X)  \notag\\
	&\xrightarrow{k\rightarrow\infty} I(Y;C|X).
	\end{align}
	
	The communication rate $R_c$ is given as
		\begin{align}\label{eq:Rc}
		\frac{k|{\cal F}_5 \cup {\cal F}_3|}{kN}& 
		= \frac{k|{\cal V}_{C}\setminus {\cal V}_{C|X}|}{Nk} 
		= \frac{|{\cal V}_{C}\setminus {\cal V}_{C|X}|}{N} \notag\\ 
		&\xrightarrow{N\rightarrow\infty} I(X;C),
		\end{align}
		whereas  $R_a$ can be written as
		\begin{align} \label{eq:Ra}
		\frac{|{\cal V}_{A|CXY}|+k|{\cal F}_8|}{kN}& \notag 
		= \frac{|{\cal V}_{A|CXY}|+k|{\cal V}_{A|C}\setminus {\cal V}_{A|CX}|}{kN} \notag\\
		&= \frac{|{\cal V}_{A|CXY}|}{kN} +\frac{|{\cal V}_{A|C}\setminus {\cal V}_{A|CX}|}{N} \notag\\
		&\xrightarrow{N\rightarrow\infty} I(A;X|C) + \frac{H(A|CXY)}{k}\notag\\
		&\xrightarrow{k\rightarrow\infty} I(A;X|C). 
		\end{align}
		
	The rates of local randomness $\rho_1$ and $\rho_2,$ respectively, are given as
			\begin{align}\label{eq:rho1}
		\rho_1=	\frac{k|{\cal F}_6|}{kN}& = \frac{k|{\cal V}_{A|CX}\setminus {\cal V}_{A|CXY}|}{Nk}=\frac{|{\cal V}_{A|CX}\setminus {\cal V}_{A|CXY}|}{N} \notag\\
			&\xrightarrow{N\rightarrow\infty} I(A;Y|CX)\quad \text{and} \\
			%\end{align}
			%\begin{equation}
                        \label{eq:rho2}
			\rho_2&=\frac{k|{V}_{Y|BC}|}{kN} \xrightarrow{N\rightarrow\infty} H(Y|BC).
			\end{align}
			
	Finally we see that conditions \eqref{equ:JSRR1}-\eqref{equ:JSRR7} are satisfied by \eqref{eq:Ro}-\eqref{eq:rho2}. Hence, given $R_a$, $R_o$, $R_c$ satisfying Theorem~\ref{Thm:JointCRR}, based on Lemma~5 and Pinsker's inequality we have  
	\begin{align} \label{eq:resolvProof}
	\mathbb E&\big[||\tilde{P}_{X_{1:k}^NY_{1:k}^N}-Q_{X^{1:kN}Y^{1:kN}}||_{{\scriptscriptstyle TV}}\big] \notag\\ 
	&\leq  \mathbb E\Big[\sqrt{2\mathbb{D}(\tilde{P}_{X_{1:k}^NY_{1:k}^N}||Q_{X^{1:kN}Y^{1:kN}})} \;\Big] \notag\\
	& \leq \sqrt{2\mathbb E \big[\mathbb{D}(\tilde{P}_{X_{1:k}^NY_{1:k}^N}||Q_{X^{1:kN}Y^{1:kN}})\big]}\mathop{\longrightarrow}^{N\rightarrow \infty} 0.
	\end{align} 	
	As a result, from \eqref{eq:resolvProof} there
        exists an $N\in\mathbb N$ for which the polar code-induced pmf
        between the pair of actions satisfies the strong coordination
        condition is given by \eqref{eq:StrngCoorCondtion}.
\end{proof}
%\addtolength{\textheight}{-12cm}
%\section{Concluding Remarks}
	% references section
	\bibliographystyle{IEEEtran}
	\bibliography{IEEEabrv,references}

\end{document}